\documentclass[12pt,reqno]{article}

\usepackage[usenames]{color}
\usepackage{amssymb}
\usepackage{amsmath}
\usepackage{amsthm}
\usepackage{amsfonts}
\usepackage{amscd}
\usepackage{graphicx}
\usepackage{diagbox}

\usepackage[colorlinks=true,
linkcolor=webgreen,
filecolor=webbrown,
citecolor=webgreen]{hyperref}

\definecolor{webgreen}{rgb}{0,.5,0}
\definecolor{webbrown}{rgb}{.6,0,0}

\usepackage{color}
\usepackage{fullpage}
\usepackage{float}

\usepackage{graphics}
\usepackage{latexsym}
\usepackage{epsf}
\usepackage{breakurl}

\newcommand{\seqnum}[1]{\href{https://oeis.org/#1}{\rm \underline{#1}}}
\def\suchthat{\, : \,}
\def\Enn{\mathbb{N}}
\DeclareMathOperator{\mex}{mex}

\begin{document}

\theoremstyle{plain}
\newtheorem{theorem}{Theorem}
\newtheorem{corollary}[theorem]{Corollary}
\newtheorem{lemma}[theorem]{Lemma}
\newtheorem{proposition}[theorem]{Proposition}

\theoremstyle{definition}
\newtheorem{definition}[theorem]{Definition}
\newtheorem{example}[theorem]{Example}
\newtheorem{conjecture}[theorem]{Conjecture}

\theoremstyle{remark}
\newtheorem{remark}[theorem]{Remark}

\title{An `Experimental Mathematics' Approach to
Stolarsky Interspersions via Automata Theory}

\author{Jeffrey Shallit\\
School of Computer Science \\
University of Waterloo\\
Waterloo, ON  N2L 3G1 \\
Canada \\
\href{mailto:shallit@uwaterloo.ca}{\tt shallit@uwaterloo.ca}}

\maketitle

\begin{abstract}
We look at the Stolarsky interspersions (such as the Wythoff array)
one more time, this time using tools from automata theory.
These tools allow easy verification of
many of the published results on these arrays, as well as proofs
of new results.
\end{abstract}

\section{Introduction}

Let the Fibonacci numbers be denoted, as usual, by
$$F_0 = 0,\ F_1 = 1,\ \text{and $F_n = F_{n-1} + F_{n-2}$ for $n \geq 2$}.$$  
We can generalize this sequence by
allowing two arbitrary starting values, to get the
generalized Fibonacci numbers:   
\begin{equation}
F^{a,b}_1 = a,\ F^{a,b}_2 = b,\ \text{ and
$F^{a,b}_n = F^{a,b}_{n-1} + F^{a,b}_{n-2}$ for $n \geq 3$.}
\label{genfib}
\end{equation}
These sequences have been studied for at least 120 years \cite{Tagiuri:1901} and
probably longer.  As Tagiuri observed, a simple induction gives a
useful formula for the $n$'th term of a generalized Fibonacci sequence
in terms of the ordinary Fibonacci sequence.
\begin{equation}
F^{a,b} (n) = a F_{n-2} + b F_{n-1} 
\end{equation}
for $n \geq 3$.

Several authors have studied {\it Stolarsky interspersions\/},
which are two-dimensional arrays 
${\bf A} = (A_{i,j})_{i,j\geq 1}$
of positive integers, strictly increasing across rows and down
columns.  The
first row $A_{1,j}$ consists of the Fibonacci numbers $(F_{j+1})_{j \geq 1}$,
and each succeeding row is a generalized Fibonacci sequence.  The
term $A_{i,1}$ in the first column is defined to be
the smallest positive integer that appeared in no preceding row.

Let us state our version of this construction.
For a set $S \subsetneq \Enn_{+} = \{1,2,3,\ldots \}$, we define
the $\mex$ function (\underline{m}inimal \underline{ex}cluded element)
$$\mex(S) = \min \{x \in \Enn_{+} \suchthat x \not\in S\};$$
then 
\begin{equation}
A_{i,1} := \mex(\{ A_{i', j} \suchthat 1\leq i' < i \text{ and } j \geq 1 \} ).
\label{mexuse}
\end{equation}

It remains to define the elements in the second column.  In the most
general case, $A_{i,2}$ is a function of $i$ and $A_{i,1}$:
$A_{i,2} = f(i, A_{i,1})$.  For the time being, we will consider the simpler case where
$f$ does not depend on $i$, and write instead $A_{i,2} = f(A_{i,1})$
for all $i \geq 1$.  We will return to the more general case
in Section~\ref{general}.
All successive terms in the
$i$'th row follow the generalized Fibonacci sequence rule:
$A_{i,n} = A_{i,n-1} + A_{i,n-2}$ for $n \geq 3$.  Thus the arrays
we study are completely determined by specifying the function $f$.

We use the notation $\alpha = (1+\sqrt{5})/2$ and
$\beta = (1-\sqrt{5})/2$.
Arrays fitting this description include 
\begin{itemize}
\item[(a)] the Wythoff array, where $f(n) = \lfloor \alpha n - \beta \rfloor$;
\item[(b)] the Stolarsky array, where $f(n) = \lfloor \alpha n + {1\over 2} 
\rfloor$; and 
\item[(c)] the dual Wythoff array, where $f(n) = \lfloor \alpha n + \beta^2
\rfloor$.
\end{itemize}
These arrays can have many interesting properties,
and are the subject of many papers;
for example, see 
\cite{Behrend:2012,Butcher:1978,Fraenkel&Kimberling:1994,Garth&Palmer:2016,Hegarty&Larsson:2006,
Hendy:1978,Kimberling:1993,Kimberling:1993b,Kimberling:1994,
Kimberling:1995,Kimberling:1995b,Kimberling:1997,
Kimberling:2010,Morrison:1980,Stolarsky:1977}.

Proofs of these properties, however, are often somewhat tedious, involving
complicated arguments
and/or lengthy inductions.  In this paper I consider the kinds of things
that can be proven using a more `experimental mathematics' approach.
In this approach we ``guess'' finite automata that can compute the
arrays, and then {\it rigorously\/} verify their correctness using
a decision procedure from first-order logic,
as implemented in the free software {\tt Walnut}.  Indeed, it often
suffices to merely state what it is you want to prove, and {\tt Walnut}
can produce a rigorous proof or disproof in a matter of seconds.
With this technique we can often, with very little work,
find explicit formulas that required
rather involved arguments.  We also solve an open question
of Kimberling from 1994, in Theorem~\ref{new}.

\section{Zeckendorf representation}

Recall that every non-negative integer $n$ has a unique representation as
a sum of distinct Fibonacci numbers $F_i$, $i \geq2$,
subject to the constraint that
no two consecutive Fibonacci numbers are used.
This is called the Fibonacci or Zeckendorf representation
\cite{Lekkerkerker:1952,Zeckendorf:1972}.   We can associate
a representation 
$ n = \sum_{1 \leq i \leq t} a_i F_{t+2-i}$
with a binary string $a_1 \cdots a_t$.  Provided the leading
digit $a_1$ equals $1$, and $(a_i, a_{i+1}) \not= (1,1)$ for
$1 \leq i < t$, this gives a unique way to associate positive
integers with binary strings.  We write $(n)_F = a_1 \cdots a_t$
for this unique binary string $a_1\cdots a_t$, and the inverse map
$[a_1\cdots a_t]_F = \sum_{1 \leq i \leq t} a_i F_{t+2-i}$.
The representation for $0$ is $\epsilon$, the empty string.
Sometimes it is useful to allow leading zeros in these binary strings;
$101, 0101, 00101,\ldots$, etc., are all considered to be the same
representation for the integer $4$.

\section{Automata, logic, and {\tt Walnut}}

A finite automaton is a simple model of a computer.  It consists of a finite
set of states, a distinguished initial state, a set of final (or accepting)
states, and labeled transitions between states.  An input string $w$ is
said to be {\it accepted} if reading the symbols of $w$, starting in the
initial state, leads to a final state; otherwise it is {\it rejected}.
For more information about this model, see, for example,
\cite{Hopcroft&Ullman:1979}.  Finite automata are often depicted
as diagrams where states are represented by circles, transition by
labeled arrows, and final states by double circles.

A finite automaton $M$ defines a subset $S \subseteq \Enn$ by letting
$S$ consist of the integers whose Zeckendorf representation 
is accepted by $M$.  For example, Figure~\ref{odd} depicts an automaton
that accepts the set of odd natural numbers.
\begin{figure}[htb]
\begin{center}
\includegraphics[width=6.5in]{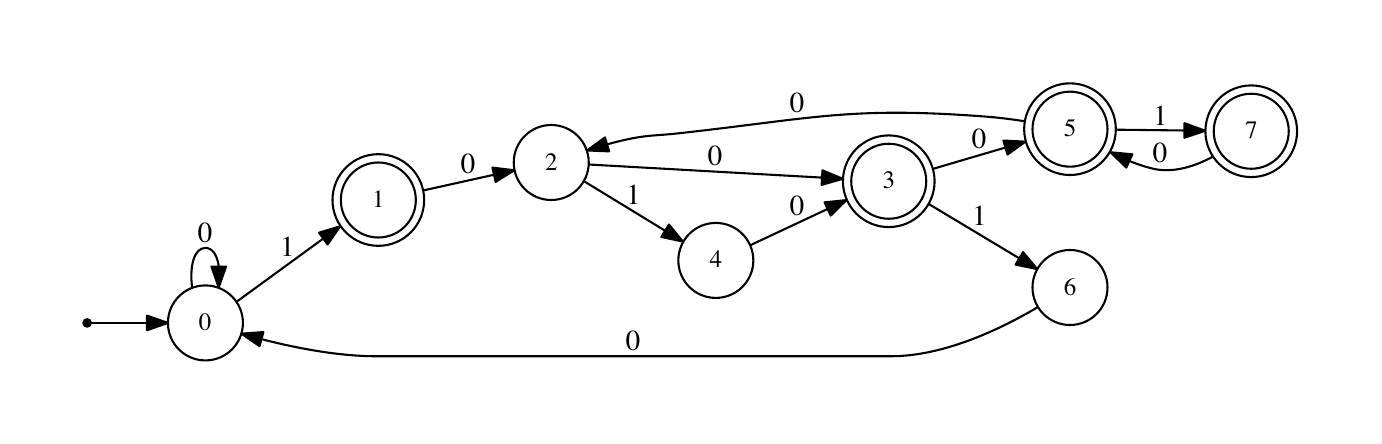}
\end{center}
\caption{Automaton accepting the set of odd natural numbers.}
\label{odd}
\end{figure} 

Finite automata can also be used to compute relations.
The way to do this is to allow transitions labeled with pairs of letters,
instead of just single letters.  Thus an input can be thought of as
representing a pair of natural numbers $(x,y)$ being read 
in parallel by the automaton, where the Zeckendorf
representation of $x$ corresponds to the first components of the input
and that of $y$ to the second components.  Since we are reading
$x$ and $y$ in parallel, we may have to allow leading zeros so that
both inputs are of the same length.  If for each $x$ there is exactly
one $y$ such that $(x,y)$ is accepted, then we can think of the automaton
as (implicitly) computing a function $f$ defined by $f(x) = y$.
For example, Figure~\ref{oddf} depicts an automaton computing
the function $f(x) = 2x+1$.
\begin{figure}[htb]
\begin{center}
\includegraphics[width=6.5in]{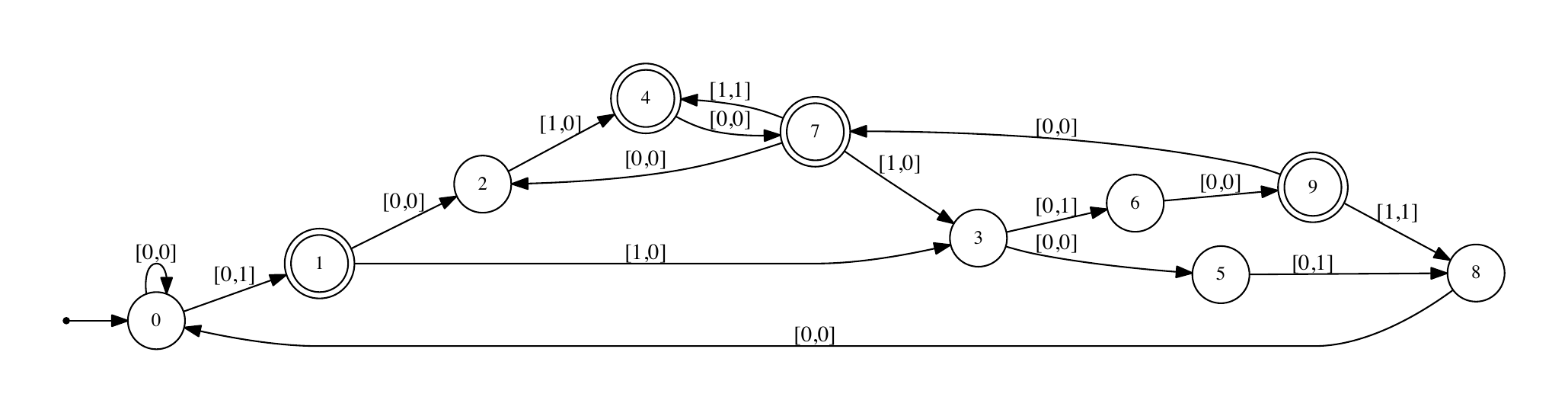}
\end{center}
\caption{Automaton for the function $f(x) = 2x+1$.}
\label{oddf}
\end{figure}
In the analogous fashion we can define a function of two variables,
$z = f(x,y)$, with triples of inputs, etc.

Now we turn to introducing {\tt Walnut}.  This is a free 
open-source software package that can rigorously prove or disprove
first-order logical statements about integer sequences and sets generated
by finite automata.  The domain of variables is assumed to be
$\Enn = \{ 0,1,2,\ldots \}$, the non-negative integers.

{\tt Walnut} uses a shorthand which is easy, but takes some getting
used to.  First, the universal quantifier $\forall$ is written
{\tt A} and the existential quantifier $\exists$ is written
{\tt E}.  Logical AND is written {\tt \&}; logical OR is written
{\tt |}; logical NOT is written {\tt \char'176}; logical
implication is written {\tt =>}; logical IFF is written
{\tt <=>}.

The command {\tt def} defines an automaton for later use, while
{\tt eval} evaluates
a first-order expression with no free variables and returns
{\tt TRUE} or {\tt FALSE}.  
The command {\tt reg} allows the user to specify an automaton through a regular
expression.  (For more about regular expressions, see
\cite{Hopcroft&Ullman:1979}.)

For example, here is the {\tt Walnut} command that
computes the automaton in Figure~\ref{odd}.
\begin{verbatim}
def odd "?msd_fib Ek n=2*k+1":
\end{verbatim}
For the automaton in Figure~\ref{oddf}, we use the command
\begin{verbatim}
def double1 "?msd_fib z=2*n+1":
\end{verbatim}
Notice that this expression {\tt z=2*n+1} has two free variables,
$z$ and $n$, so the produced automaton has transitions on pairs
of letters.  The order of inputs is alphabetical order on the
variable names when the command is defined; in this example,
$n$ is spelled out by the first coordinates of inputs and
$z$ is spelled out by the second coordinates.

\begin{lemma}
Suppose $x$ is a binary string, and $m,n$ are integers
such that $m = [x]_F$ and $n = [x0]_F$.  Then $m+n = [x00]_F$.
\label{lem4}
\end{lemma}

\begin{proof}
Write $x = a_1 \cdots a_t$; then we have
$[x]_F = \sum_{1\leq i \leq t} a_i F_{t+2-i}$
and
$[x0]_F = \sum_{1\leq i \leq t} a_i F_{t+3-i}$.
Then 
$$m+n = \sum_{1\leq i \leq t} a_i (F_{t+2-i} + F_{t+3-i})
= \sum_{1\leq i \leq t} a_i F_{t+4-i} = [x00]_F.$$
\end{proof}

\begin{lemma}
Let $a,b$ be integers with $a \geq 0$ and
$b = \lfloor \alpha a \rfloor + \epsilon$
for $\epsilon \in \{ 0, 1 \}$.  Define the generalized Fibonacci
sequence $F^{a,b}$ by \eqref{genfib}. Then, if $a+b>0$, we have
\begin{equation}
(F^{a,b} (n))_F = (F^{a,b} (3))_F \, 0^{n-3}
\label{fab}
\end{equation}
for $n \geq 3$.
\label{usef}
\end{lemma}

\begin{proof}
We first build a `shift' automaton that accepts inputs $x$ and
$y$ iff $y = x0$.  We can do this in {\tt Walnut} with a regular
expression:
\begin{verbatim}
reg shift msd_fib msd_fib "([0,0]|[0,1][1,1]*[1,0])*":
\end{verbatim}
This produces the two-state automaton in Figure~\ref{shift}.
\begin{figure}[htb]
\begin{center}
\includegraphics[width=4in]{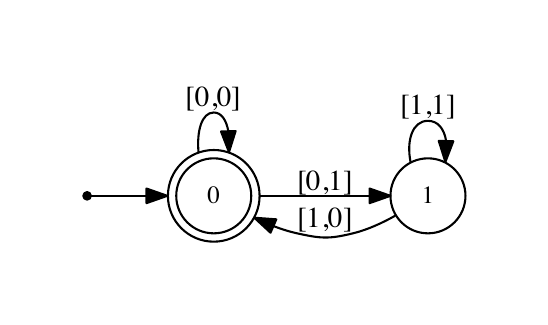}
\end{center}
\caption{Automaton for the shift.}
\label{shift}
\end{figure}

It suffices to show that $(F^{a,b} (4))_F = (F^{a,b} (3))_F \, 0$, because then
\eqref{fab} follows by induction from Lemma~\ref{lem4}.
We define $b = \lfloor \alpha a \rfloor + \epsilon$,
$c = a+b$, and $d = b+c$.
Then we need to show that $(d)_F = (c)_F \, 0$.
\begin{verbatim}
eval fab "?msd_fib Aa,b,c,d,x ($phin(a,x) & (b=x|b=x+1) & c=a+b
   & d=b+c) => $shift(c,d)":
\end{verbatim}
And {\tt Walnut} returns {\tt TRUE}.
Here {\tt phin} is an automaton of two inputs $n$ and $x$ that
accepts if and only if $x = \lfloor \alpha n \rfloor$.
It can be found, for example, in \cite[p.~278]{Shallit:2023}.
\end{proof}

It would be very nice if there were automata taking three inputs 
$i,j,x$ in parallel for which $A_{i,j} = x$ for the arrays $\bf A$ 
we discuss here.  (Indeed, this nice property is true for the arrays in
\cite{Shallit:2025}.)  Unfortunately, for the arrays we discuss
in this paper, this is not possible, because
while columns grow linearly, rows grow exponentially, and this violates
a growth rate condition on synchronized sequences
\cite{Shallit:2021}.  Depending on the choice of $f$, we can (in some cases)
find automata computing the $i$'th element of any individual column
$A_{i,n}$.   For the arrays we discuss in this paper, it turns out
that this is enough!
Lemma~\ref{usef} ensures that there is an 
an integer $n_0$ such that the set $S$ of all 
values appearing in columns $\geq n_0$ are expressible in a simple
way from those in column $n_0$, and hence (depending on $f$) there can
be an automaton accepting $S$.  A similar result holds
for all columns of even (resp., odd) index.
As we will see, this suffices to prove
many existing theorems about the arrays $\bf A$.

For more about {\tt Walnut}, see \cite{Mousavi:2016,Shallit:2023} and
the web page \\
\centerline{\url{https://cs.uwaterloo.ca/~shallit/walnut.html} .}

\section{The experimental approach to Stolarsky interspersions}

The experimental approach to Stolarsky interspersions involves seven steps.
\begin{enumerate}
\item Compute many terms of the first column of the array ${\bf A} = (A_{i,j})_{i, j\geq 1}$
from its definition. In practice, for the arrays we discuss here,
500 to 5000 terms seem sufficient.

\item Use a procedure based on a variation of the Myhill-Nerode theorem to
``guess'' an automaton $M$ computing a function $i \rightarrow A'_{i,1}$
matching the known data; see \cite[\S 5.6]{Shallit:2023}.  We now hope
that $A'_{i,1} = A_{i,1}$, but of course this has to be rigorously verified!

\item Verify, using {\tt Walnut},
that $M$ truly does compute a function $A'_{i,1}$ (that is, for every $i \geq 1$
the automaton accepts exactly one input of the form $(i,x)$).

\item Verify that $A'_{i,1}$ is strictly increasing as a function of $i$.

\item Compute, using {\tt Walnut},
an automaton for $A'_{i,2} = f(A'_{i,1})$ and
$A'_{i,3} = A'_{i,1} + A'_{i,2}$.

\item From this, using
Lemma~\ref{usef}, compute an automaton accepting the Zeckendorf representations
of the set $ S := \{ A'_{i,j} \suchthat i \geq 1, j \geq 3 \}$.

\item Verify inductively that $M$ is correct by proving that
$A'_{i,1} = \mex( \{ A'_{i',j} \suchthat i'<i, j \geq 1 \} )$.
Then $A'_{i,n} = A_{i,n}$ for all $i$ and $n=1,2$.

\end{enumerate}

\section{The Wythoff array}

In 1980 Morrison studied an infinite array, now called the Wythoff array,
with interesting properties \cite{Morrison:1980}.
The first few rows and columns are given in Table~\ref{tab1}.
\begin{table}[htb]
\begin{center}
\begin{tabular}{c||c|c|c|c|c|c|c|c|c|c|c}
\diagbox{$i$}{$j$} & 1& 2& 3& 4& 5& 6& 7& 8& 9&10& $\cdots$\\
\hline
1&   1&   2&   3&   5&   8&  13&  21&  34&  55&  89&  \\
  2&   4&   7&  11&  18&  29&  47&  76& 123& 199& 322& \\
  3&   6&  10&  16&  26&  42&  68& 110& 178& 288& 466& \\
  4&   9&  15&  24&  39&  63& 102& 165& 267& 432& 699& \\
  5&  12&  20&  32&  52&  84& 136& 220& 356& 576& 932& \\
   6&  14&  23&  37&  60&  97& 157& 254& 411& 665&1076&\\
   7&  17&  28&  45&  73& 118& 191& 309& 500& 809&1309&\\
   8&  19&  31&  50&  81& 131& 212& 343& 555& 898&1453&\\
   9&  22&  36&  58&  94& 152& 246& 398& 644&1042&1686&\\
  10&  25&  41&  66& 107& 173& 280& 453& 733&1186&1919&\\
  $\vdots$
\end{tabular}
\end{center}
\caption{The Wythoff array $W_{i,j}$.}
\label{tab1}
\end{table}

The Wythoff array corresponds to the choice $f(n) = \lfloor \alpha n - \beta \rfloor$.  (Morrison used a different definition, which we show to be equivalent
below.)

Let us now carry out the seven steps of the experimental approach.

Step 1: We compute the first 500 terms of column 1.

\medskip

Step 2:  These suffice to guess a 10-state automaton
{\tt w1}
that we hope computes $W_{n,1}$ for $n \geq 1$.  It is depicted in
Figure~\ref{col1}.
\begin{figure}[htb]
\begin{center}
\includegraphics[width=6.5in]{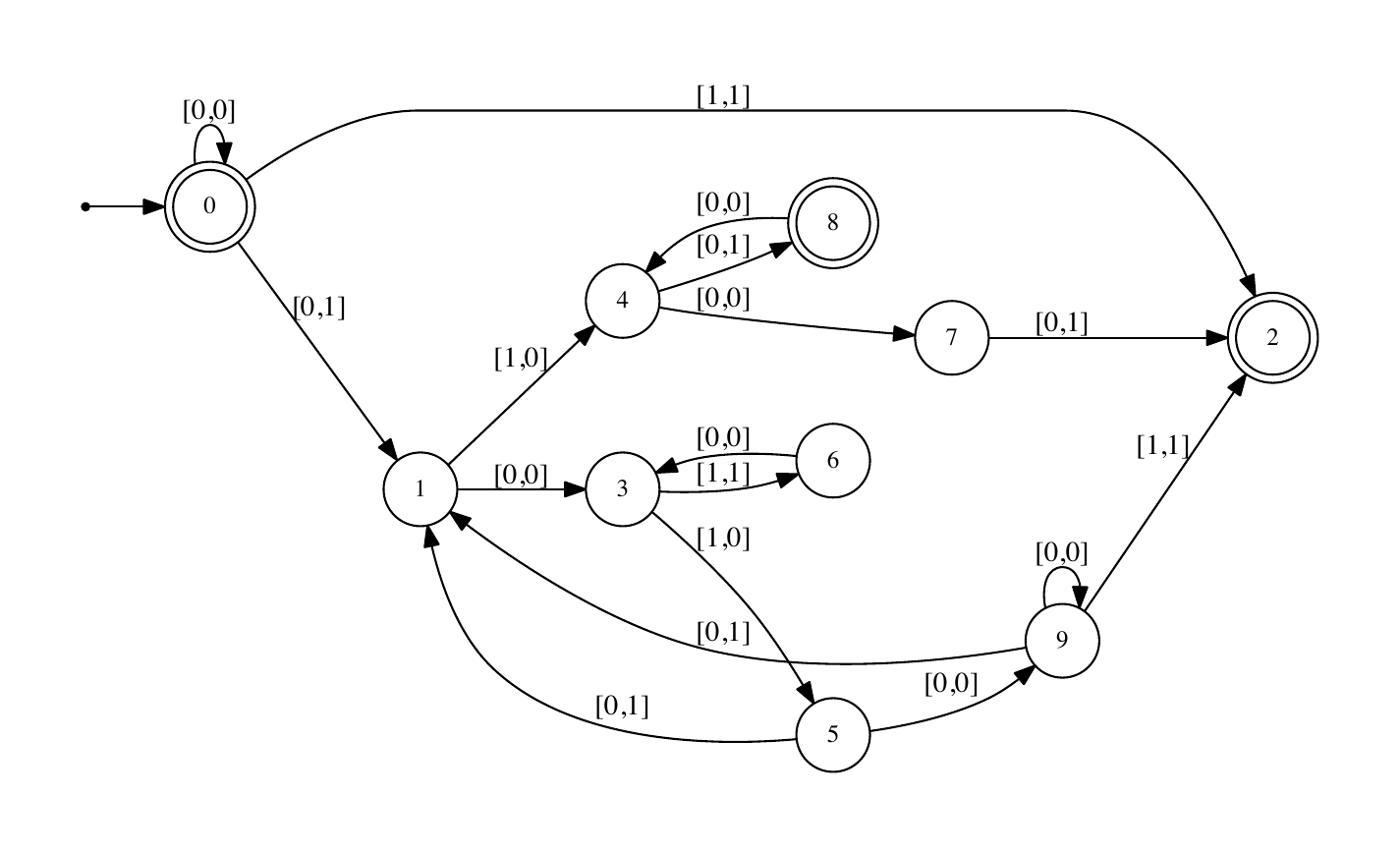}
\end{center}
\caption{Guessed automaton for the function $n \rightarrow W_{n,1}$.}
\label{col1}
\end{figure} 

\medskip

Step 3:  We verify that {\tt w1} actually computes a function.
We do this as follows:
\begin{verbatim}
eval w1_func1 "?msd_fib An (n>=1) => Ex $w1(n,x)":
eval w1_func2 "?msd_fib ~En,x,y n>=1 & x!=y & $w1(n,x) & $w1(n,y)":
\end{verbatim}
And both return {\tt TRUE}.

\medskip

Step 4:  We verify that {\tt w1} computes a strictly increasing
function.  We do this with
\begin{verbatim}
eval increasing_w "?msd_fib An,x,y (n>=1 & $w1(n,x) & $w1(n+1,y)) => x<y":
\end{verbatim}

Step 5:  We find automata for computing columns 2 and 3 from the definitions:
\begin{verbatim}
def fw "?msd_fib Ey $phin(n+1,y) & z+1=y":
def w2 "?msd_fib Ex $w1(i,x) & $fw(x,z)":
def w3 "?msd_fib Ex,y $w1(i,x) & $w2(i,y) & z=x+y":
\end{verbatim}
These have 13 and 18 states, respectively.
The automaton {\tt w2} is depicted in Figure~\ref{col2}.
\begin{figure}[htb]
\begin{center}
\includegraphics[width=6.5in]{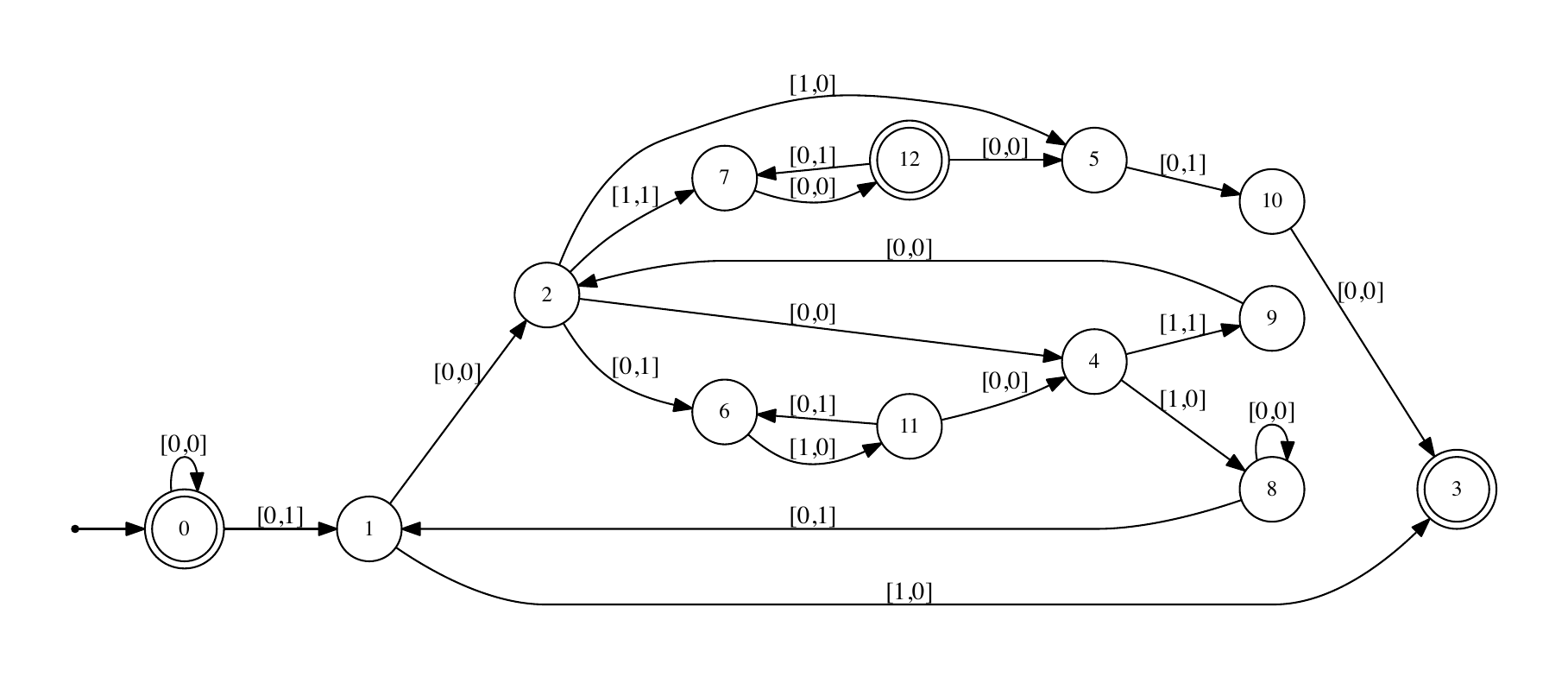}
\end{center}
\caption{Automaton computing the function $n \rightarrow W_{n,2}$.}
\label{col2}
\end{figure}

\medskip

Step 6:  We compute an automaton testing membership in 
$S = \{ W_{i,j} \suchthat i\geq 1, j\geq 2 \}$.  To do so, we use
the fact that $(W_{i,j})_F = (W_{i,3})_{F} 0^{j-3}$ for $j \geq 3$.  So we first
make an automaton accepting the 
set $\{ W_{i,3} \suchthat i \geq 1 \}$, and then, from that
automaton, we concatenate the language $0^*$ on its right using
the {\tt concat} command of {\tt Walnut}.  The {\tt alphabet}
command forces the correct numeration system to be attached
to the result.
\begin{verbatim}
def wc3_mem "?msd_fib Ei i>=1 & $w3(i,n)":
reg all0 msd_fib "0*":
concat wcols3 wc3_mem all0:
alphabet wcols msd_fib $wcols3:
\end{verbatim}
Then we form the automaton for $S$ by also including the values
of $\{ W_{i,2} \suchthat i \geq 1 \}$.
\begin{verbatim}
def sw "?msd_fib (Ei i>=1 & $w2(i,n)) | $wcols(n)":
\end{verbatim}

Step 7:
Now we are poised to prove that the first column computed by {\tt w1} is
correct, by induction on $i$.  The base case, $i = 1$, is clear.
We use {\tt Walnut} to check that if the ``mex'' property holds
for rows $1,2,\ldots, i-1$,
then it also holds for row $i$.

It suffices to show
that $W_{i,1}$ is the least number not in
$T_i := S \, \cup \, \{ W_{i',1} \suchthat 1 \leq i' < i \}$.
\begin{verbatim}
def chk1 "?msd_fib (~$sw(n)) & Aip,x (ip<i & $w1(ip,x)) => n!=x":
# decide if x doesn't appear in set T_i
def mex "?msd_fib $chk1(i,z) & Ay (y>=1 & $chk1(i,y)) => y>=z":
# decide if z is the smallest of all of those x not appearing
eval mex_exists "?msd_fib Ai (i>=2) => Ex $mex(i,x)":
# check that there is indeed a mex for all i
eval chk "?msd_fib Ai,x (i>=2 & $mex(i,x)) => $w1(i,x)":
\end{verbatim}
And {\tt Walnut} returns {\tt TRUE} for both of the last two commands.
We have now (almost) completed the induction proof that {\tt w1} correctly
computes the array's first column.
However, there is one small technical point.
The astute reader will notice that the ``mex'' result proved here is
that $W_{i,1}$ is the least number not in 
$T'_i := \{ W_{i',1} \suchthat 1 \leq i' < i \} \, \cup \, \{ W_{j,k} \suchthat
j \geq 1, \ k \geq 2 \}$, which is slightly different from the definition
\eqref{mexuse}
mentioned earlier:  the set $T'_i$ includes {\it all\/} of the elements  
in columns $2$ and higher, not just those in rows $i' < i$.  Clearly
$T_i \subseteq T'_i$.  However, we
now argue this makes no difference.  For if $W_{i,1}$ does not occur
in $T_i$ and 
does occur in $T'_i$,
then we would have $W_{i,1} = W_{j,k}$
for some $j \geq i$ and $k \geq 2$.   But this is impossible since
the array $\bf W$ is strictly increasing across rows and down columns.
Now the proof is complete.

Morrison \cite{Morrison:1980} defined the first two columns of the
Wythoff array in a different way.  Namely, he used
the definition
$W_{i,1} = \lfloor \alpha \lfloor \alpha i \rfloor \rfloor$ and
$W_{i,2} = \lfloor \alpha^2 \lfloor \alpha i \rfloor \rfloor$.  
We can prove this is equivalent to ours, as follows:
\begin{verbatim}
def mor1 "?msd_fib Ey $phin(n,y) & $phin(y,z)":
def phi2n "?msd_fib Ex $phin(n,x) & z=x+n":
def mor2 "?msd_fib Ey $phin(n,y) & $phi2n(y,z)":
eval check_mor1 "?msd_fib An,x $mor1(n,x) <=> $w1(n,x)":
eval check_mor2 "?msd_fib An,x $mor2(n,x) <=> $w2(n,x)":
eval check_m1 "?msd_fib An,z1,z2 (n>=1 & $mor1(n,z1) & $mor2(n,z2)) 
   => $fw(z1,z2)":
\end{verbatim}
which returns {\tt TRUE}.  Here {\tt phi2n} computes $\lfloor \alpha^2 n \rfloor$. 

We can also prove closed forms for columns 1 and 2.
\begin{proposition}
\leavevmode
\begin{itemize}
\item[(a)] $W_{i,1} = \lfloor \alpha^2 i \rfloor - 1$ for $i \geq 1$.
\item[(b)] $W_{i,2} = 2 \lfloor \alpha i \rfloor + i-1$ for $i \geq 1$.
\end{itemize}
\label{prop3}
\end{proposition}

\begin{proof}
\leavevmode
\begin{verbatim}
eval prop5a "?msd_fib Ai,x,y (i>=1 & $w1(i,x) & $phi2n(i,y)) => x+1=y":
eval prop5b "?msd_fib Ai,x,y (i>=1 & $w2(i,x) & $phin(i,y)) => x+1=2*y+i":
\end{verbatim}
\end{proof}

Another property of this array was proven by Morrison \cite{Morrison:1980},
although its statement as written was imprecise.
More precisely stated, it is the following:
\begin{theorem}
$$ \{ W_{i,2} - W_{i,1} \suchthat i \geq 1 \} =
\{ W_{j,2k+1} \suchthat j\geq 1 \text{ and }  k \geq 0 \}.$$
\end{theorem}

\begin{proof}
To prove the theorem, we construct automata testing membership in
each set, and then check that they accept exactly the same sets.

An automaton for the set on the left is easy to construct:
\begin{verbatim}
def left "?msd_fib Ei,x,y $w1(i,x) & $w2(i,y) & z+x=y":
\end{verbatim}

For the set on the right, we use the fact that 
$(W_{j,\ell})_F = (W_{j,3})_F 0^{\ell -3}$ for
$\ell \geq 3$.  Thus we can construct the set 
$\{ W_{j,2k+1} \suchthat j\geq 1 \text{ and }  k \geq 0 \}$
as follows:
\begin{verbatim}
def wc3_memb "?msd_fib Ei i>=1 & $w3(i,n)":
reg even0 msd_fib "(00)*":
concat w3even wc3_memb even0:
alphabet w3e msd_fib $w3even:
def right "?msd_fib (Ej j>=1 & $w1(j,n)) | $w3e(n)":
\end{verbatim}
And now we check equality:
\begin{verbatim}
eval checkeq "?msd_fib An (n>=1) => ($left(n) <=> $right(n))":
\end{verbatim}
And {\tt Walnut} returns {\tt TRUE}.
\end{proof}

Finally, we can prove Kimberling's beautiful
formula \cite[p.~314]{Kimberling:1994}
for the entry $W_{i,j}$.  It is also recorded
in the entry for sequence \seqnum{A035513} in \cite{Sloane:2026}.
\begin{proposition}
For $i, j \geq 1$ we have $W(i,j) = \lfloor \alpha i \rfloor F_{j+1} +
(i-1) F_j$.
\label{bott1}
\end{proposition}

\begin{proof}
By induction on $j$.  The base cases are $j=1$ and $j=2$, and
follow from Proposition~\ref{prop3}.  Now assume $j \geq 3$.
We have
\begin{align*}
W(i,j) &= W(i,j-1) + W(i,j-2) \\
&= \lfloor \alpha i \rfloor F_j + (i-1) F_{j-1} +
\lfloor \alpha i \rfloor F_{j-1} + (i-1) F_{j-2} \\
&= \lfloor \alpha i \rfloor (F_j + F_{j-1}) +
(i-1) (F_{j-1} + F_{j-2}) \\
&= \lfloor \alpha i \rfloor F_{j+1} + (i-1) F_j,
\end{align*}
as desired.
\end{proof}

\section{The Stolarsky array}

The Stolarsky array \cite{Stolarsky:1977}
actually predates the Wythoff array, and shares many of
the same properties.  It is displayed in Table~\ref{tab3}.
\begin{table}[htb]
\begin{center}
\begin{tabular}{c||c|c|c|c|c|c|c|c|c|c|c}
\diagbox{$i$}{$j$} & 1& 2& 3& 4& 5& 6& 7& 8& 9&10& $\cdots$\\
\hline 
 1&   1&   2&   3&   5&   8&  13&  21&  34&  55&  89&  \\
  2&   4&   6&  10&  16&  26&  42&  68& 110& 178& 288& \\
  3&   7&  11&  18&  29&  47&  76& 123& 199& 322& 521& \\
  4&   9&  15&  24&  39&  63& 102& 165& 267& 432& 699& \\
  5&  12&  19&  31&  50&  81& 131& 212& 343& 555& 898& \\
   6&  14&  23&  37&  60&  97& 157& 254& 411& 665&1076&\\
   7&  17&  28&  45&  73& 118& 191& 309& 500& 809&1309&\\
   8&  20&  32&  52&  84& 136& 220& 356& 576& 932&1508&\\
   9&  22&  36&  58&  94& 152& 246& 398& 644&1042&1686&\\
  10&  25&  40&  65& 105& 170& 275& 445& 720&1165&1885& \\
  $\vdots$
\end{tabular}
\end{center}
\caption{The Stolarsky array $S_{i,j}$.}
\label{tab3}
\end{table}

As mentioned before,
the function underlying the Stolarsky array is $f(n) = \lfloor \alpha n + 
{1\over 2} \rfloor$.  Let us carry out the seven steps for this array.

Step 1:  We can easily compute the first $1000$ terms
of column $1$ from the formula.

Step 2:  We
guess an automaton {\tt s1} for it based on this data.
Our guessed automaton is {\tt s1} and appears in Figure~\ref{s1}.
\begin{figure}[htb]
\begin{center}
\includegraphics[width=6.5in]{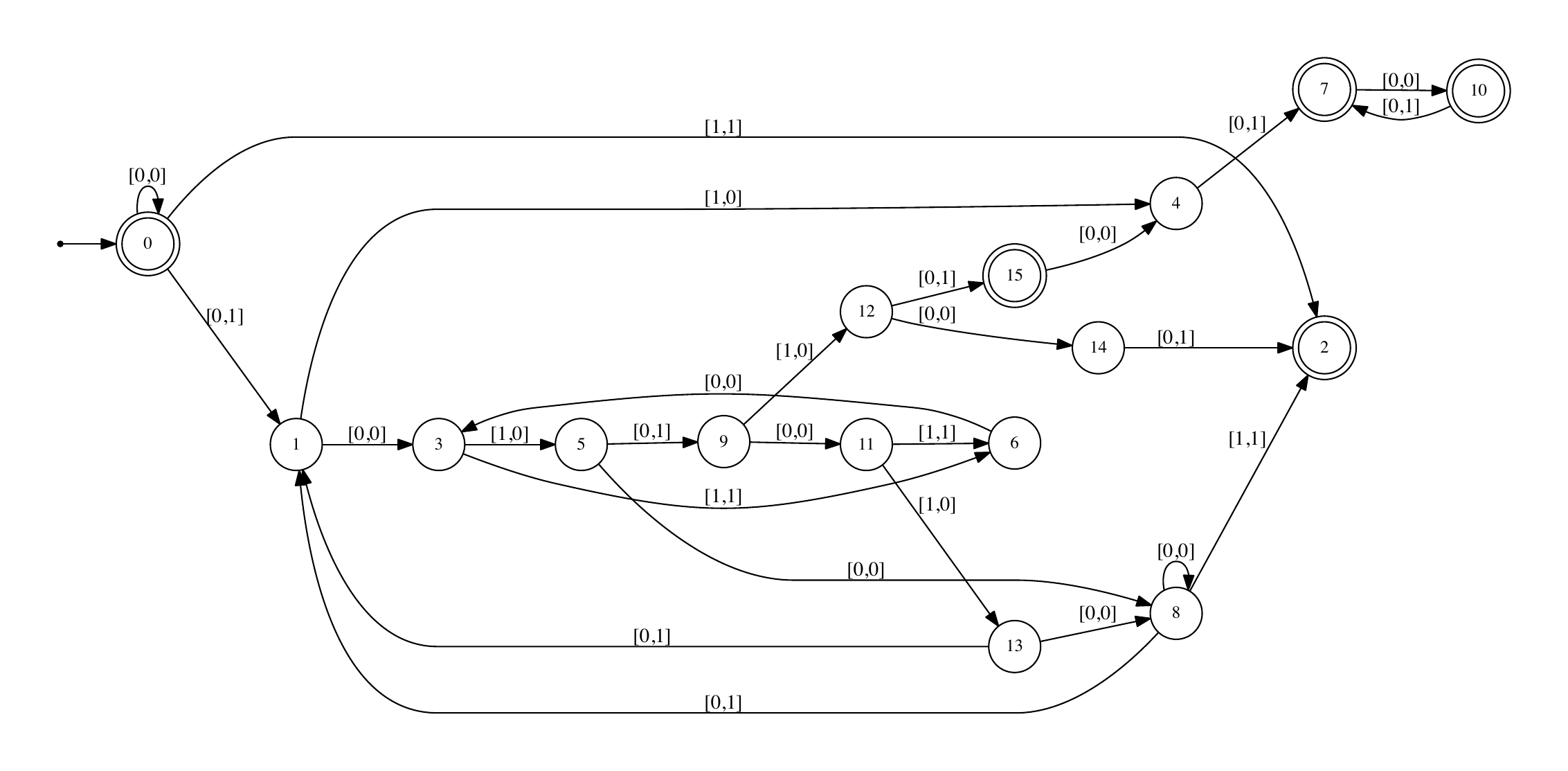}
\end{center}
\caption{Guessed automaton for $S_{i,1}$.}
\label{s1}
\end{figure}

Step 3:  We verify that {\tt s1} actually computes a function.
\begin{verbatim}
eval s1_func1 "?msd_fib An (n>=1) => Ex $s1(n,x)":
eval s1_func2 "?msd_fib ~En,x,y n>=1 & x!=y & $s1(n,x) & $s1(n,y)":
\end{verbatim}

Step 4:  We verify that {\tt s1} is strictly increasing.
\begin{verbatim}
eval increasing_s "?msd_fib An,x,y (n>=1 & $s1(n,x) & $s1(n+1,y)) => x<y":
\end{verbatim}

Step 5: 
We obtain automata for the second and third columns.
For the second column we use the relation
$S_{i,2} = f(S_{i,1})$.
Notice that 
$$\left\lfloor \alpha n + {1\over 2} \right\rfloor
= \left\lfloor {{{2\alpha n} + 1} \over 2} \right\rfloor 
= \left\lfloor {{\lfloor 2 \alpha n \rfloor + 1}\over 2} \right\rfloor,
$$
so we can define {\tt s2} as follows:
\begin{verbatim}
def s2 "?msd_fib En,y $s1(i,n) & $phin(2*n,y) & z=(y+1)/2":
\end{verbatim}
For the third column we use
the relation $S_{i,3} = S_{i,1} + S_{i,2}$.
\begin{verbatim}
def s3 "?msd_fib Ex,y $s1(i,x) & $s2(i,y) & z=x+y":
\end{verbatim}

Step 6:  
Lemma~\ref{usef} gives the Zeckendorf expansion of all 
columns $S'_{i,n}$; namely,
$(S'_{i,n})_F = (S'_{i,3})_F 0^{n-3}$ for all $n \geq 3$.

\medskip

Step 7:  Our last goal is to prove that $S_{i,1} = S'_{i,1}$, for then
the correctness of all other entries follow from above.
We need to verify
that $S'_{i,1}$ is the least positive integer not appearing in previous
rows.  We do this in the same way we did for the Wythoff array:
\begin{verbatim}
def sc3_mem "?msd_fib Ei i>=1 & $s3(i,n)":
reg alls0 msd_fib "0*":
concat scol3 sc3_mem alls0:
alphabet scols msd_fib $scol3:
def ssw "?msd_fib (Ei i>=1 & $s2(i,n)) | $scols(n)":
def chks1 "?msd_fib (~$ssw(n)) & Aj,x (j<i & $s1(j,x)) => n!=x":
def mexs "?msd_fib $chks1(i,x) & Ay (y>=1 & $chks1(i,y)) => y>=x":
eval mexs_exists "?msd_fib Ai (i>=2) => Ex $mexs(i,x)":
eval chks2 "?msd_fib Ai,x (i>=2 & $mexs(i,x)) => $s1(i,x)":
\end{verbatim}

It now follows that our guess for $S_{i,1}$ was correct.
As an application, we can recover simple expressions for the first
and second columns found by Hendy \cite[Thm.~3]{Hendy:1978}.
\begin{theorem}
We have $S_{n,1} = \lfloor n \alpha^2 - \alpha/2 \rfloor$ for $n\geq 1$.
\label{thm1}
\end{theorem}

\begin{proof}
Simple manipulation shows
$$ \lfloor n \alpha^2 - \alpha/2 \rfloor =
n + \left\lfloor {{\lfloor (2n-1)\alpha \rfloor } \over 2} \right\rfloor .$$
So we can prove Theorem~\ref{thm1} as follows:
\begin{verbatim}
eval thms1 "?msd_fib An,x (n>=1) =>
   ($s1(n,x) <=> Ey,z $phin(2*n-1,y) & z=y/2 & x=z+n)":
\end{verbatim}
\end{proof}

\begin{theorem}
We have $S_{n,2} = \lfloor n(\sqrt{5} + 2) - \alpha \rfloor$
for $n\geq 1$.
\label{thm2}
\end{theorem}

\begin{proof}
Since $\sqrt{5} + 2 = 2 \alpha + 1$, we can prove the theorem as follows:
\begin{verbatim}
eval thms2 "?msd_fib An,x (n>=1) => ($s2(n,x) <=> Ey $phin(2*n-1,y) & x=y+n)":
\end{verbatim}
\end{proof}

Stolarsky conjectured that 
$$S_{n,2} - S_{n,1} \in
\{ S_{i,1} \suchthat 1 \leq i <n \} \cup \{ S_{i,2} \suchthat
1 \leq i < n  \}$$
for $n \geq 2$.
This conjecture was proved by 
Butcher \cite{Butcher:1978}
and Hendy \cite{Hendy:1978}.
We can prove this as follows:
\begin{verbatim}
eval stol_conjecture "?msd_fib An,x,y (n>=2 & $s1(n,x) & $s2(n,y))
=> Ei,a,b 1<=i & i<n & $s1(i,a) & $s2(i,b) & (y=x+a|y=x+b)":
\end{verbatim}
The reader may enjoy comparing the simplicity of this proof
to the proofs of Butcher and Hendy.

Furthermore we can determine when $S_{n,2}  - S_{n,1}  = S_{i,1} $:
\begin{verbatim}
def which "?msd_fib n>=1 & Ex,y,i,a 1<=i & i<n & $s1(n,x) & $s2(n,y) & 
$s1(i,a) & y=x+a":
\end{verbatim}

From this we can recover another result of Hendy \cite[Thm.~5]{Hendy:1978}.
\begin{proposition}
Let $n \geq 1$.  There exists $i$ such that $S_{n,2} - S_{n,1} = S_{i,1} $ if and only
if $n = \lfloor k\alpha + 1-\alpha/2 \rfloor$ for some $k\geq 1$.
\end{proposition}

\begin{proof}
\leavevmode
\begin{verbatim}
eval cond "?msd_fib An (n>=1) => ((Ei,x,y,z $s1(n,x) & $s2(n,y) 
   & $s1(i,z) & y=x+z) <=> (Ek,t $phin(2*k-1,t) & n=t/2+1))":
\end{verbatim}
And {\tt Walnut} returns {\tt TRUE}.
\end{proof}

Let us use this approach to solve a problem of Kimberling
\cite{Kimberling:1994}.
Given a Stolarsky interspersion $(A_{i,j})_{i,j\geq 1}$,
he defined the {\it classification sequence}
$\delta_i := A_{i,2} - \lfloor \alpha A_{i,1} \rfloor$, and proved
that $\delta_i \in \{ 0, 1 \}$.  Some arrays, such
as the Wythoff array, have very simple classification
sequences; for the Wythoff array it is $\delta_i = 1$.  Kimberling asked for
a ``convenient formula'' for the classification sequence
$\delta_i$ for the Stolarsky array $\bf s$.

Using {\tt Walnut}, and automata for the first two columns of a Stolarsky
interspersion, we can easily compute an automaton {\tt delta}
that accepts $(n,t)$ if $\delta_n = t$.
\begin{verbatim}
def delta "?msd_fib Ex,y,z $s1(n,x) & $phin(x,y) & $s2(n,z) & z=y+t":
\end{verbatim}
It is depicted in Figure~\ref{class}.
\begin{figure}[htb]
\begin{center}
\includegraphics[width=6.5in]{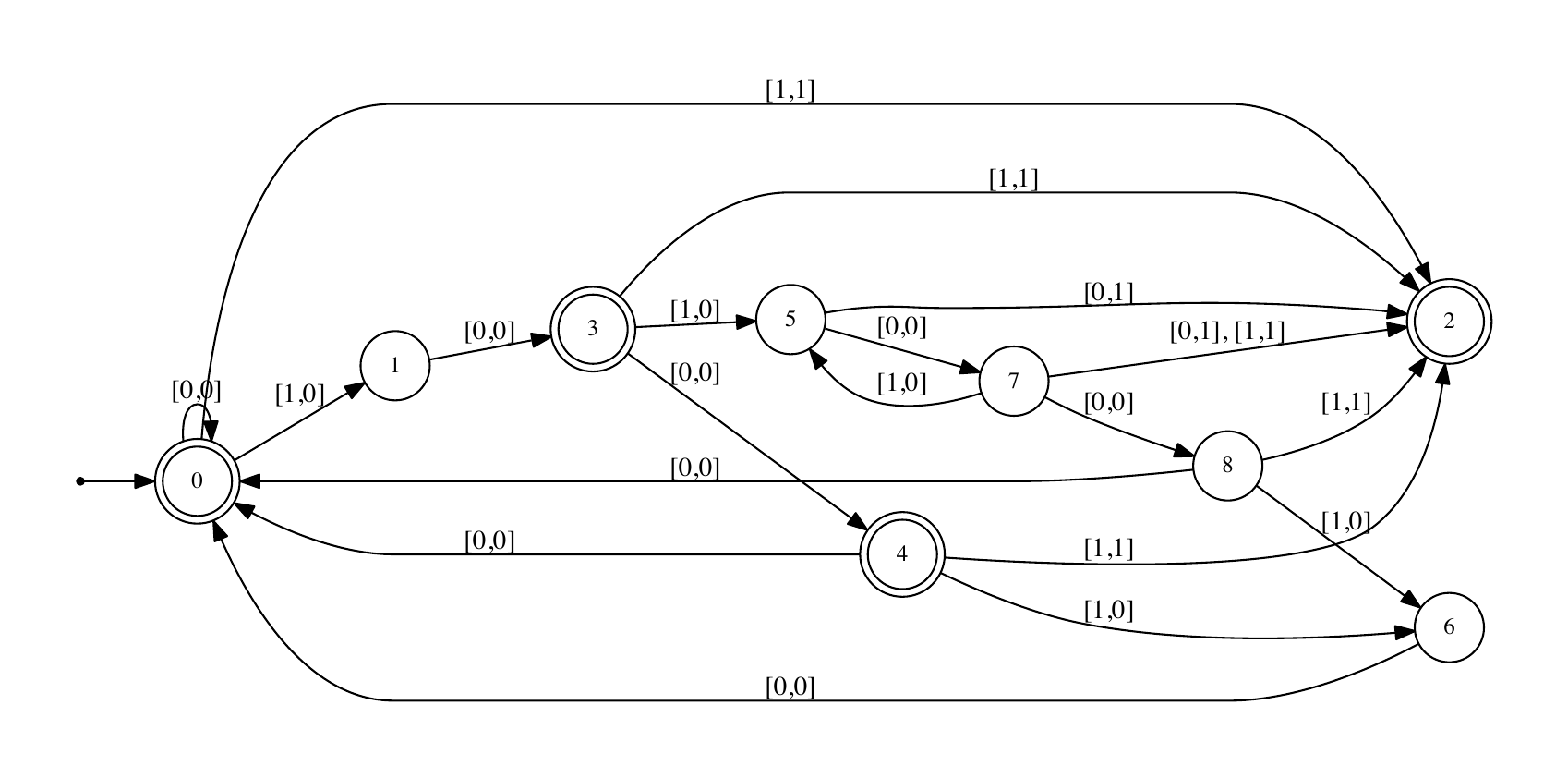}
\end{center}
\caption{Automaton for $(\delta_i)_{i \geq 1}$.}
\label{class}
\end{figure}

Of course, some would object to an automaton proffered as a ``convenient
formula''.  Nevertheless, once we have the automaton, we can efficiently
compute an arbitrary term of the sequence $\delta_i$.  Let us take an
experimental approach to determining another kind of formula
for $(\delta_i)_{i \geq 1}$.  

For example, we can look at the frequency of occurrences of the symbol
$1$ in the first $n$ terms of $(\delta_i)_{i \geq 1}$, and we discover
it seems to be $1/2$.  This is reminiscent of a Rote sequence
\cite{Rote:1994}, so it suggests verifying this by computing subword complexity of
the sequence, where subword complexity is the function mapping $n$ to the
number of distinct length-$n$ blocks occurring in the sequence.  
The empirical data suggests the subword complexity is $2n$ for $n \geq 1$,
which confirms the possibility of $(\delta_i)_{i \geq 1}$ as a Rote
sequence.  This, in turn, suggests that perhaps $\delta_i = 1$
if and only if $\{ \alpha n + \rho \} \in [0, \kappa)$ for some
numbers $\rho, \gamma$.   An experimental approach now tries pairs
$(\rho, \gamma)$ generated at random and seeing how close the resulting
sequence is to $(\delta_i)_{i \geq 1}$.  We quickly converge on
the candidates $\rho = 1/2, \gamma = (3-\alpha)/2$.  Once we have this
guess, we can verify it using {\tt Walnut}:
\begin{verbatim}
def l1 "?msd_fib Ex $phin(2*n-1,x) & z=(x+3)/2":
# automaton computing floor(gamma + alpha*n)
def l2 "?msd_fib Ex $phin(2*n-1,x) & z=(x/2)+2":
# automaton computing floor(gamma+ alpha*n + 1/2)
def l3 "?msd_fib Ex,y $l1(n,x) & $l2(n,y) & x=y":
# automaton deciding if fractional part of gamma+alpha*n < 1/2
eval tmp "?msd_fib An (n>=1) => ($l3(n) <=> $delta(n,1))":
# check they are the same, result is TRUE
\end{verbatim}
We have now proved the following new result, which offers an answer
to Kimberling's question.
\begin{theorem}
Let $\gamma = (3-\alpha)/2$ and $\rho = 1/2$.  We have 
$$ \delta_n = \begin{cases}
1, & \text{if $\{ \alpha n + \gamma \} \in [0,\rho)$}; \\
0, & otherwise.
\end{cases}
$$
\label{new}
\end{theorem}

As an application, we can prove a formula for the Stolarsky array,
analogous to that given in Proposition~\ref{bott1} for the Wythoff array.
This appears to be new.
\begin{theorem}
For $i, j \geq 1$ we have
$S(i,j) = \lfloor \alpha (i-1/2) + 1/2 \rfloor F_{j-1} + (\lfloor \alpha (i-1/2) \rfloor + i) F_j$.
\end{theorem}

\begin{proof}
By induction on $j$.  The base cases are
$S(i,1) = \lfloor \alpha(i-1/2) \rfloor + i$
and
$S(i,2) = \lfloor \alpha (i-1/2) + 1/2 \rfloor + \lfloor \alpha(i-1/2) \rfloor + i$
for $i \geq 1$, and they follow from Theorems~\ref{thm1} and \ref{thm2}.
The rest of the induction follows in the same
manner as in the proof of Proposition~\ref{bott1}.
\end{proof}

\section{The dual Wythoff array}

The dual Wythoff array ${\bf D} = (D_{i,j})_{i,j\geq 1}$ was introduced by
Kimberling \cite[p.~314]{Kimberling:1994}.  
It corresponds to the case
$f(n) = \lfloor \alpha n + \beta^2 \rfloor$, although
Kimberling did not define it like that.  The first few rows and columns
are given in Table~\ref{dual}.
\begin{table}[htb]
\begin{center}
\begin{tabular}{c||c|c|c|c|c|c|c|c|c|c|c}
\diagbox{$i$}{$j$} & 1& 2& 3& 4& 5& 6& 7& 8& 9&10& $\cdots$\\
\hline 
1&   1&   2&   3&   5&   8&  13&  21&  34&  55&  89&  \\
  2&   4&   6&  10&  16&  26&  42&  68& 110& 178& 288& \\
  3&   7&  11&  18&  29&  47&  76& 123& 199& 322& 521& \\
  4&   9&  14&  23&  37&  60&  97& 157& 254& 411& 665& \\
  5&  12&  19&  31&  50&  81& 131& 212& 343& 555& 898& \\
   6&  15&  24&  39&  63& 102& 165& 267& 432& 699&1131&\\
   7&  17&  27&  44&  71& 115& 186& 301& 487& 788&1275&\\
   8&  20&  32&  52&  84& 136& 220& 356& 576& 932&1508&\\
   9&  22&  35&  57&  92& 149& 241& 390& 631&1021&1652&\\
  10&  25&  40&  65& 105& 170& 275& 445& 720&1165&1885&\\
  $\vdots$
\end{tabular}
\end{center}
\caption{The dual Wythoff array $D_{i,j}$.}
\label{dual}
\end{table}

We can prove properties of this array just like we did with 
Stolarsky's array.  

Step 1:  We compute the first 1000 entries of column $1$
using $f(n) = \lfloor \alpha n + \beta^2 \rfloor$.

Step 2:
We guess an automaton {\tt d1} for the
first column from empirical data.  It has $11$ states
and is depicted in
Figure~\ref{dualf}.
\begin{figure}[htb]
\begin{center}
\includegraphics[width=6.5in]{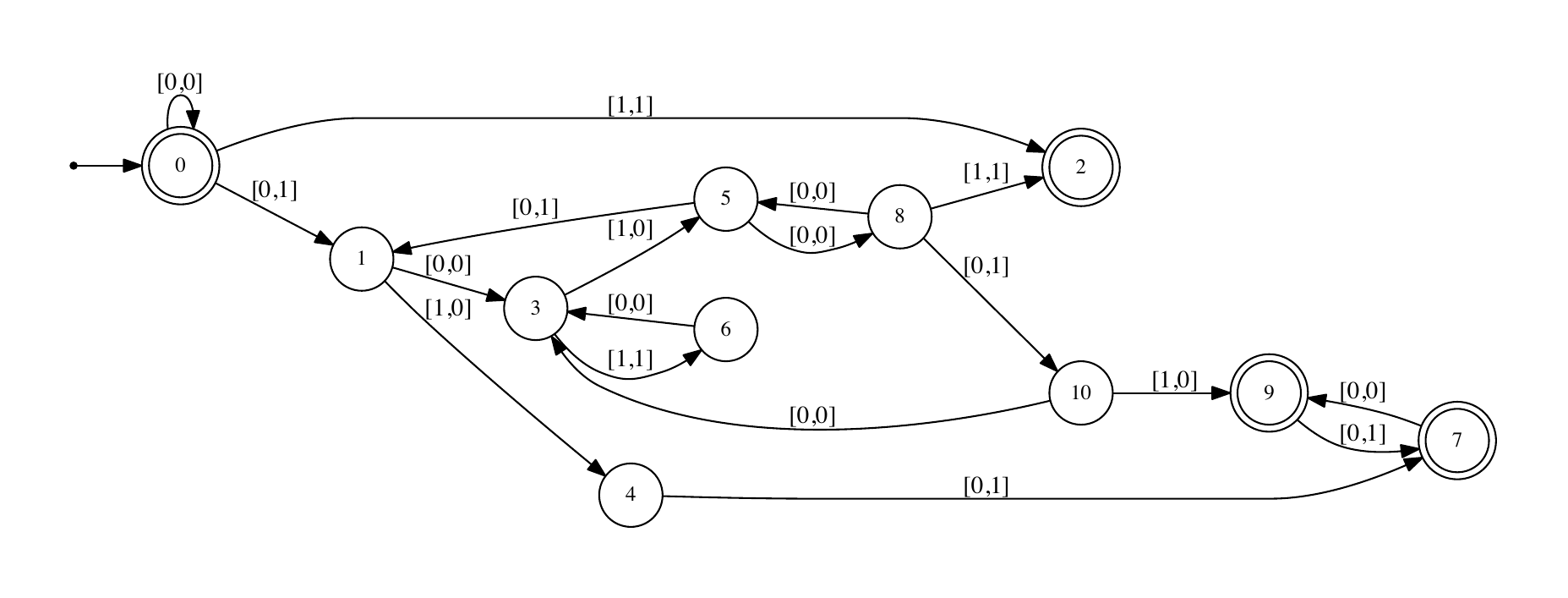}
\end{center}
\caption{Guessed automaton for $D_{i,1}$.}
\label{dualf}
\end{figure}

Step 3:
We verify that {\tt d1} computes a function.
\begin{verbatim}
eval d1_func1 "?msd_fib An (n>=1) => Ex $d1(n,x)":
eval d1_func2 "?msd_fib ~En,x,y n>=1 & x!=y & $d1(n,x) & $d1(n,y)":
\end{verbatim}

Step 4:  We verify that {\tt d1} is strictly increasing.
\begin{verbatim}
eval increasing_d "?msd_fib An,x,y (n>=1 & $d1(n,x) & $d1(n+1,y)) => x<y":
\end{verbatim}

Step 5:  We compute columns $2$ and $3$:
\begin{verbatim}
def fd "?msd_fib Ex $phin(n-1,x) & z=x+2":
def d2 "?msd_fib Ex $d1(i,x) & $fd(x,z)":
def d3 "?msd_fib Ex,y $d1(i,x) & $d2(i,y) & z=x+y":
\end{verbatim}

Step 6:  We compute the set $S$.
\begin{verbatim}
def dc3_mem "?msd_fib Ei i>=1 & $d3(i,n)":
reg all0 msd_fib "0*":
concat dcol dc3_mem all0:
alphabet dcols msd_fib $dcol:
def dw "?msd_fib (Ei i>=1 & $d2(i,n)) | $dcols(n)":
\end{verbatim}

Step 7:  We prove the induction step that the `mex' property holds for the
array $D$.
\begin{verbatim}
def chkd1 "?msd_fib (~$dw(n)) & Aj,x (j<i & $d1(j,x)) => n!=x":
def mexd "?msd_fib $chkd1(i,x) & Ay (y>=1 & $chkd1(i,y)) => y>=x":
eval chkd2 "?msd_fib Ai,x (i>=2 & $mexd(i,x)) => $d1(i,x)":
\end{verbatim}
Now we know that the automaton {\tt d1} indeed computes the first
column correctly.  Furthermore we can find a closed form
for it.

\begin{theorem}
We have $D_{i,1} = \lceil \alpha^2 i - \alpha \rceil$.
\end{theorem}

\begin{proof}
It suffices to prove that
$D_{i,1} = \lfloor \alpha^2 i - \alpha \rfloor + 1$ for $i \geq 2$.
Note that $\lfloor \alpha^2 i - \alpha \rfloor =
\lfloor (\alpha + 1) i - \alpha \rfloor
= i+\lfloor \alpha(i-1) \rfloor$.
\begin{verbatim}
eval chkd1 "?msd_fib Ai,x,y (i>=2 & $d1(i,x) & $phin(i-1,y)) => x=y+i+1":
\end{verbatim}
\end{proof}

Kimberling used a different description of the array.  We now prove
his description is equivalent to ours.
\bigbreak
\begin{theorem}
\leavevmode
\begin{itemize}
\item[(a)] The classification sequence $\delta_i$ for the
dual Wythoff array is given by $\delta_1 = 1$ and
$\delta_i = 0$ for all $i \geq 2$.
\item[(b)] We have
$$ D_{i,1} = \begin{cases}
\lfloor i \alpha \rfloor + i, & \text{if there exists $k\geq 1$
such that $i = \lfloor k \alpha \rfloor + k + 1$ }; \\
\lfloor i \alpha \rfloor + i-1, & \text{otherwise.}
\end{cases} 
$$
\end{itemize}
\end{theorem}

\begin{proof}
\leavevmode
\begin{verbatim}
def delta "?msd_fib Ex,y,t $d1(i,x) & $d2(i,y) & $phin(x,t) & z+t=y":
eval checka "?msd_fib Ai (i>=2) => $delta(i,0)":
def cond "?msd_fib Ek,x k>=1 & $phin(k,x) & i=x+k+1":
def d1_alt "?msd_fib ($cond(i) & Ex $phin(i,x) & z=x+i) |
(~$cond(i) & Ex $phin(i,x) & z+1=x+i)":
eval checkb "?msd_fib Ai,x,y (i>=1 & $d1(i,x) & $d1_alt(i,y)) => x=y":
\end{verbatim}
\end{proof}

\section{More general interspersions}
\label{general}

Up to now the second column of the arrays we study has been a function only
of the values appearing in the first column.  Now we consider a more general
case where the entries in the second column $A_{i,2}$ are functions 
of the first column and index,
as follows:  $A_{i,2} = \lfloor \alpha A_{i,1} \rfloor 
+ \delta_i$.   Here $\delta_i$ is the so-called classification sequence
of Kimberling \cite{Kimberling:1994}.

\subsection{The EFC array}

Kimberling defined an array, which he called the EFC array (\underline{e}ven
\underline{f}irst \underline{c}olumn),
as follows:  ${\bf E} = (E_{i,j})_{i,j\geq 1}$ with classification
sequence defined by $\delta_1 = 1$ and $\delta_{2k} = 1$, $\delta_{2k+1} = 0$
for all $k \geq 1$.  The first few rows and columns are given in
Table~\ref{efctab}.
\begin{table}[htb]
\begin{center}
\begin{tabular}{c||c|c|c|c|c|c|c|c|c|c|c}
\diagbox{$i$}{$j$} & 1& 2& 3& 4& 5& 6& 7& 8& 9&10& $\cdots$\\
\hline
1&   1&   2&   3&   5&   8&  13&  21&  34&  55&  89&  \\
  2&   4&   7&  11&  18&  29&  47&  76& 123& 199& 322& \\
  3&   6&   9&  15&  24&  39&  63& 102& 165& 267& 432& \\
  4&  10&  17&  27&  44&  71& 115& 186& 301& 487& 788& \\
  5&  12&  19&  31&  50&  81& 131& 212& 343& 555& 898& \\
   6&  14&  23&  37&  60&  97& 157& 254& 411& 665&1076&\\
   7&  16&  25&  41&  66& 107& 173& 280& 453& 733&1186&\\
   8&  20&  33&  53&  86& 139& 225& 364& 589& 953&1542&\\
   9&  22&  35&  57&  92& 149& 241& 390& 631&1021&1652&\\
  10&  26&  43&  69& 112& 181& 293& 474& 767&1241&2008&\\
  $\vdots$
\end{tabular}
\end{center}
\caption{The EFC array $E_{i,j}$.}
\label{efctab}
\end{table}

With this definition, Kimberling gave an explicit
characterization of the first column $E_{2n,1}$.
However,
his proof was rather complicated, consisting of three pages
and six cases.  We can handle this rather easily using our experimental
approach.  

Step 1:  we compute the first 1000 terms of the first column of the
EFC array.

Step 2:
We guess an automaton {\tt efc1} for the first column from the data;
it has 33 states, too big to display here.  This automaton, as well as all
others mentioned in this paper, can be downloaded from
\url{https://cs.uwaterloo.ca/~shallit/Papers/stol-automata.tar}.

Step 3: We verify that {\tt efc1} computes a function.
\begin{verbatim}
eval efc1_func1 "?msd_fib An (n>=1) => Ex $efc1(n,x)":
eval efc1_func2 "?msd_fib ~En,x,y n>=1 & x!=y & $efc1(n,x) & $efc1(n,y)":
\end{verbatim}

Step 4: We verify that the function computed by {\tt efc1} is increasing.
\begin{verbatim}
eval increasing_efc "?msd_fib An,x,y (n>=1 & $efc1(n,x) & $efc1(n+1,y)) => x<y":
\end{verbatim}

Step 5:  We compute the second column
from Kimberling's description of $\delta_i$,
we get the second column, and the third column from the Fibonacci recurrence.
\begin{verbatim}
def efc2 "?msd_fib (i=1&z=2)| (i>=2 & Ex,y,k,r $efc1(i,x) & $phin(x,y) & 
   i=2*k+r & r<=1 & z=y+(1-r))":
def efc3 "?msd_fib Ex,y $efc1(i,x) & $efc2(i,y) & z=x+y":
\end{verbatim}
These have $47$ and $67$ states, respectively.

Step 6:  We compute the set $S$.
\begin{verbatim}
def efcg3 "?msd_fib Ei i>=1 & $efc3(i,n)":
reg all0 msd_fib "0*":
concat efc_zero3 efcg3 all0:
alphabet efc_cols3 msd_fib $efc_zero3:
def efc_cols2 "?msd_fib (Em m>=1 & $efc2(m,n))|$efc_cols3(n)":
\end{verbatim}

Step 7:  We check the `mex' property.
\begin{verbatim}
def chk_efc1 "?msd_fib (~$efc_cols2(n)) & Aj,x (j<i & $efc1(j,x)) => n!=x":
def mex_efc "?msd_fib $chk_efc1(i,x) & Ay (y>=1 & $chk_efc1(i,y)) => y>=x":
eval chk_efc "?msd_fib Ai,x (i>=2 & $mex_efc(i,x)) => $efc1(i,x)":
\end{verbatim}
and {\tt Walnut} returns {\tt TRUE}.

Now that we know the array is correct,
we can prove Kimberling's characterization of column 1:
\begin{theorem}
For $n \geq 1$ we have
\begin{align*}
E_{2n,1} &= 2 \lfloor n \alpha \rfloor + 2n ; \\
E_{2n+1,1} &= 2 \lfloor n \alpha \rfloor + 2n+2. 
\end{align*}
Hence $E_{n,1}$ is even for all $n \geq 2$.
\end{theorem}

\begin{proof}
\leavevmode
\begin{verbatim}
eval efc_even "?msd_fib An,x,y (n>=1 & $efc1(2*n,x) & $phin(n,y)) 
   => x=2*y+2*n":
eval efc_odd "?msd_fib An,x,y (n>=1 & $efc1(2*n+1,x) & $phin(n,y)) 
   => x=2*y+2*n+2":
\end{verbatim}
\end{proof}

\subsection{The ESC array}

Kimberling \cite{Kimberling:1994} defined a Stolarsky interspersion
where $\delta_i = i \bmod 2$ for $i \geq 1$; he called the resulting
array ``ESC'' (\underline{e}ven \underline{s}econd \underline{c}olumn);
We will denote it as $(E'_{i,j})_{i, j \geq 1}$.
He conjectured that with this choice of $\delta$, the second column is
always even, a fact later proved by
Behrend \cite{Behrend:2012}.  We can reprove this result and also
determine an automaton for the second column sequence.

Step 1:  We compute 1000 terms of the first column from the definition.

Step 2:  We deduce an automaton {\tt esc1} that we hope is correct for $E'_{i,1}$.
It has 39 states.  We do not display it here because it is too large.

Step 3:  We verify that {\tt esc1} represents a function.
\begin{verbatim}
eval esc1_func1 "?msd_fib An (n>=1) => Ex $esc1(n,x)":
eval esc1_func2 "?msd_fib ~En,x,y n>=1 & x!=y & $esc1(n,x) & $esc1(n,y)":
\end{verbatim}

Step 4:  We verify that the function computed is increasing.
\begin{verbatim}
eval increasing_esc "?msd_fib An,x,y (n>=1 & $esc1(n,x) & $esc1(n+1,y)) => x<y":
\end{verbatim}

Step 5:  We compute automata for the second and third columns.
\begin{verbatim}
def esc2 "?msd_fib Ex,y,r $esc1(i,x) & $phin(x,y) & i=r+2*(i/2) & z=y+r":
def esc3 "?msd_fib Ex,y $esc1(i,x) & $esc2(i,y) & z=x+y":
\end{verbatim}
These have $52$ and $72$ states, respectively.

Step 6:  We form the automaton for the set $S$.
\begin{verbatim}
def escg3 "?msd_fib Ei i>=1 & $esc3(i,n)":
reg all0 msd_fib "0*":
concat esc_zero3 escg3 all0:
alphabet esc_cols3 msd_fib $esc_zero3:
def esc_cols2 "?msd_fib (Em m>=1 & $esc2(m,n))|$esc_cols3(n)":
\end{verbatim}

Step 7:
We check the `mex' property.
\begin{verbatim}
def chk_esc1 "?msd_fib (~$esc_cols2(n)) & Aj,x (j<i & $esc1(j,x)) => n!=x":
def mex_esc "?msd_fib $chk_esc1(i,x) & Ay (y>=1 & $chk_esc1(i,y)) => y>=x":
eval mex_esc_exists "?msd_fib Ai (i>=2) => Ex $mex_esc(i,x)":
eval chk_esc "?msd_fib Ai,x (i>=2 & $mex_esc(i,x)) => $esc1(i,x)":
\end{verbatim}

The second column sequence begins
$2, 6, 12, 14, 18, 24, 28, 32, 36, 40, \ldots$
and is OEIS sequence \seqnum{A380835} in \cite{Sloane:2026}.

We can check the assertion that the second column is always even, as
follows:
\begin{verbatim}
eval tmp "?msd_fib Ai,x (i>=1 & $esc2(i,x)) => Ek x=2*k":
\end{verbatim}

\section{Final remarks}

In principle, 
one can use this method to compute automata and prove theorems about
Stolarsky interspersions based on other functions $f$ and other
classification sequences.  However, currently we have no proof that,
for example, there will always be an automaton for the first
column of a Stolarsky interspersion with
periodic classification sequence $\delta_i$, nor a direct way to
compute the automaton from $\delta_i$ without ``guessing''.
We also do not know how the number of states
grows with the period of the periodic sequence.
It would be interesting to try to prove something about this.

For example, for the classification sequence
defined by $\delta_i = 1$ if $i \equiv 1$ (mod $3$), and
$0$ otherwise, we can get an automaton {\tt k100} for the
first column sequence
$$ 1, 4, 7, 9, 12, 14, 17, 20, 23, 25, 27, 30, 33, 35, 38, 40, 44, 46, 49, \ldots $$
with 87 states.  With our method we can prove, for example,
the following result:
\begin{theorem}
The Stolarsky array with classification
sequence $1,0,0,1,0,0,1,0,0,\ldots$ has a third
column with terms that are only $0, 1$ (mod $3$), never
$2$ (mod $3$).
\end{theorem}

\begin{proof}
We verify our guess and the property.  The {\tt Walnut} code is given
without commentary, as it is just the same in spirit as before.
\begin{verbatim}
eval k100_func1 "?msd_fib An (n>=1) => Ex $k100(n,x)":
eval k100_func2 "?msd_fib ~En,x,y n>=1 & x!=y & $k100(n,x) & $k100(n,y)":
eval increasing_k100 "?msd_fib An,x,y (n>=1 & $k100(n,x) & $k100(n+1,y)) => x<y":
def is1mod3 "?msd_fib (z=1 & Ek n=3*k+1)|(z=0 & ~Ek n=3*k+1)":
def k1002 "?msd_fib Ex,y,r $k100(i,x) & $phin(x,y) & $is1mod3(i,r) & z=y+r":
def k1003 "?msd_fib Ex,y $k100(i,x) & $k1002(i,y) & z=x+y":
def k100g3 "?msd_fib Ei i>=1 & $k1003(i,n)":
reg all0 msd_fib "0*":
concat k100_zero3 k100g3 all0:
alphabet k100_cols3 msd_fib $k100_zero3:
def k100_cols2 "?msd_fib (Em m>=1 & $k1002(m,n))|$k100_cols3(n)":
def chk_k100 "?msd_fib (~$k100_cols2(n)) & Aj,x (j<i & $k100(j,x)) => n!=x":
def mex_k100 "?msd_fib $chk_k100(i,x) & Ay (y>=1 & $chk_k100(i,y)) => y>=x":
eval chk_k100 "?msd_fib Ai,x (i>=2 & $mex_k100(i,x)) => $k100(i,x)":
eval col3 "?msd_fib Ai,x (i>=1 & $k1003(i,x)) => ~Ek x=3*k+2":
\end{verbatim}
And {\tt Walnut} returns {\tt TRUE}.
\end{proof}

One thing I do not currently know how to do is use these kinds of methods
to discuss the so-called ``row sequence'', which is sequence
$r_n$ defined by the row number in which $n$ appears, say, in the
Wythoff array.  I do not see a way to study this sequence using automata
theory.

One can object that, for example, for the {\tt ESC} array, we have only
verified that the second column is always even, but we have not explained
why.  This raises the question of what would be considered a satisfying
answer.  Evidently a simple formula of the form $\lfloor \gamma n + \rho
\rfloor$ for some irrationals $\gamma, \rho$ would be acceptable, but
considering the
subword complexity of the first-difference sequence of \seqnum{A380835}---i.e., $4, 6, 2, 4, 6, 4, 4, 4, 4, \ldots$---precludes the
existence of such a formula.
We would need something more complicated.
Perhaps the 52-state automaton for {\tt esc2} is going to be the simplest
``formula'' we can find.  We note that Behrend did not prove any kind
of simple formula for the second column of {\tt ESC}, either.

\end{document}